\documentclass[submission,copyright,creativecommons]{eptcs}
\providecommand{\event}{LSFA 2012} 
\usepackage{breakurl}             

\usepackage{amssymb}

\newtheorem{theorem}{Theorem}

\newtheorem{corollary}[theorem]{Corollary}

\newtheorem{definition}[theorem]{Definition}

\newtheorem{lemma}[theorem]{Lemma}

\newtheorem{proposition}[theorem]{Proposition}

\newenvironment{proof}[1][Proof]{\noindent\textbf{#1.} }{\ \rule{0.5em}{0.5em}}

\title{Sequent Calculi for the classical fragment of Bochvar and Halld\'en's Nonsense Logics}
\author{Marcelo E. Coniglio \qquad\qquad Mar\'ia I. Corbal\'{a}n
\institute{Centre for Logic, Epistemology and the History of Science\\
and\\
Department of Philosophy\\
State University of Campinas, Brazil}
\email{\quad coniglio@cle.unicamp.br \quad\qquad  inescorbalan@yahoo.com.ar}
}
\def\titlerunning{Sequent Calculi for Bochvar and Halld\'en's Nonsense Logics}
\def\authorrunning{M. E. Coniglio \& M. I. Corbal\'{a}n}
\begin{document}
\maketitle

\begin{abstract}
In this paper sequent calculi for the  classical  fragment (that is, the conjunction-disjunction-implica\-tion-negation fragment) of the nonsense logics \textbf{B}$_{3}$, introduced by  Bochvar, and \textbf{H}$_{3}$, introduced by Halld\'en, are presented. 
These calculi are obtained by  restricting in an appropriate way the application of the rules of a sequent calculus for classical propositional logic  \textbf{CPL}.  The nice symmetry between the provisos in the rules reveal the semantical  relationship between these logics.  The Soundness and Completeness theorems for both calculi are obtained, as well as the respective Cut elimination theorems.
\end{abstract}

\section*{Introduction}

The study of logical paradoxes from a formal perspective has produced several proposals in the literature. In particular, 3-valued propositional logics were proposed in which, besides the two `classical' truth-values, the third one plays the role of  a `nonsensical'  or `meaningless' truth value. This is why these logics are known as `logics of nonsense'. In 1938 (\cite{boch:38})   A. Bochvar introduced the first logic of nonsense, by means of 3-valued logical matrices. Since the nonsensical truth value is not distinguished, Bochvar's logic is paracomplete but it is not paraconsistent: the negation $\lnot$ is explosive (from a contradiction everything follows) but the third-excluded law does not hold. In 1949 S. Halld\'en (\cite{hall:49})  proposed a closely related logic of nonsense by means of 3-valued logical matrices in which the third truth-value is distinguished, producing a paraconsistent, non-paracomplete logic.

Both logics share the same main feature: the nonsensical truth-value is `infectious' in the sense that, given a valuation $v$, every formula having at least one propositional variable with nonsensical truth-value under $v$ also gets the non-sensical truth-value under $v$.  Also, both logics contain, besides the connectives $\lnot$ for negation and $\wedge$ conjunction,  an unary connective which allows to recover all the classical inferences (cf.~\cite{con:cor:11,cor:12}).

The respective `classical' fragments of each of these two logics (that is, the $\{\lnot,\vee,\wedge,\to\}$-fragments)  are interesting since they together  constitute the only two possibilities for extending the usual matrices of  classical logic with a third nonsensical,  `infectious' truth-value $\frac12$: either $\frac12$ is designated or it is not. The former corresponds to the  `classical' fragment of Halld\'en's logic, while the  latter corresponds to the  same fragment of Bochvar's  logic. It is not hard to establish, by semantical means, a relationship between these two fragments and classical logic: given a classically valid  inference $\Gamma \vdash \alpha$ over the  language generated by $\{\lnot,\vee,\wedge,\to\}$, if the propositional variables ocurring in $\Gamma$ also occur in $\alpha$ then $\Gamma \vdash \alpha$   is valid in Halld\'en's logic \textbf{H}$_{3}$. Dually, if the propositional variables ocurring in $\alpha$ also occur in $\Gamma$ then such classically valid inference is valid in Bochvar's logic \textbf{B}$_{3}$. This duality  is a direct consequence of the criterion adopted in each logic with respect to the third truth-value (namely, designated vs. non-designated), and the fact that this non-sensical truth-value propagates through any complex formula.  Since $\{\vee, \to\}$ and $\{\wedge,\to\}$ can be defined as usual from  $\{\lnot,\wedge\}$ and $\{\lnot,\vee\}$, respectively, the observation above can also be applied to the $\{\lnot,\wedge\}$ and  $\{\lnot,\vee\}$-fagments of both logics.

This paper introduced two cut-free sequent calculi for the $\{\lnot,\vee,\wedge,\to\}$-fragment of each logic of nonsense mentioned above. Both systems are obtained by imposing restrictions on the rules of the usual sequent calculus for classical propositional logic  \textbf{CPL}. In the calculus for the classical fragment of  Halld\'en's logic, the introduction rules for conjunction, implication and negation on the left side of the sequent are restricted. In the calculus for the fragment of  Bochvar's logic the restriction is imposed to the introduction rule for disjunction, implication and  negation on the right side.  In this manner, the  relationship  between classical logic and both logics  became explicit through restrictions on the rules for the logical connectives $\lnot$, $\wedge$, $\vee$ and $\to$.

\section{Preliminaries} \label{prelim}

Along this paper, we fix a denumerable set  $prop$  of propositional variables, as well as three propositional  signatures: $\Sigma_1$ just containing a negation (unary) connective $\lnot$ and a disjunction (binary) connective $\vee$; $\Sigma_2$ just containing negation $\lnot$ and a conjunction (binary) connective $\wedge$; and $\Sigma_0$, containing $\lnot$, $\vee$, $\wedge$, and an implication (binary) connective $\to$. The set of formulas generated by $\Sigma_i$ and $prop$ will be denoted by $For_i$, for $i=0,1,2$. 
The disjunction $\vee$ and the implication $\rightarrow$ are defined in $For_2$ as  $\alpha \vee \beta =_{def}\lnot \left( \lnot \alpha \wedge \lnot \beta
\right) $ and $\alpha \to \beta =_{def}\lnot \left( \alpha \wedge \lnot
\beta \right) $, respectively. By its turn, the conjunction $\wedge$ and the implication $\to$ are defined in $For_1$ as  $\alpha \wedge \beta =_{def}\lnot \left( \lnot \alpha \vee \lnot \beta
\right) $ and $\alpha \to \beta =_{def}\lnot  \alpha \vee \beta$, respectively.

For $i=0,1,2$, the function $var:For_i \rightarrow \wp \left( prop\right) $ which
assigns to each formula the set of propositional variables appearing in it  is defined recursively
as usual.  
When $\Gamma \subseteq For_i$ is a set of formulas then $var\left( \Gamma \right)
=\bigcup_{\gamma \in \Gamma} var\left( \gamma \right)$.

The next step is to recall a well-known cut-free sequent calculus for classical propositional logic  \textbf{CPL} defined over the signature $\Sigma_0$.

\begin{definition}
By a \emph{sequent} \emph{S} over $\Sigma_i$ ($i=0,1,2$) we shall mean an ordered pair $\left\langle
\Gamma ,\Delta \right\rangle $ of (non-simultaneously empty) finite sets of
formulas in $For_i$.
\end{definition}

 We shall use the more suggestive notation $\Gamma \Rightarrow
\Delta $ for the sequent $\left\langle \Gamma ,\Delta \right\rangle $. Sequents of the form $\left\langle
\Gamma ,\emptyset \right\rangle$, $\left\langle
\emptyset ,\Delta \right\rangle$, $\left\langle
\Gamma ,\{\alpha\} \right\rangle$ and  $\left\langle
\{\alpha\} ,\Delta \right\rangle$  will be denoted by   $\Gamma \Rightarrow$, \ $\Rightarrow \Delta $, \ 
 $\Gamma \Rightarrow \alpha $ and $\alpha \Rightarrow \Delta $, respectively. As usual, we write $\alpha ,\Gamma$ (or $\Gamma, \alpha$) and $\alpha , \beta, \Gamma$ (or $\Gamma, \alpha,\beta$) instead of 
$\Gamma \cup \left\{ \alpha \right\}$ and $\Gamma \cup \left\{ \alpha,\beta \right\}$, respectively.

\begin{definition} \label{seqC}

The sequent calculus \textbf{C}  over $\Sigma_0$  is defined as follows:

\paragraph{Axioms}
\[
\textrm{Ax }\textrm{ }%
\displaystyle \frac{}{\alpha \Rightarrow \alpha}
\]

\paragraph{Structural rules}

\begin{center}
\bigskip 
\[
\textrm{W }\mathbf{\Rightarrow }\textrm{ }%
\begin{tabular}{c}
$\Gamma \Rightarrow \Delta $ \\ \hline
$\alpha ,\Gamma \Rightarrow \Delta $%
\end{tabular}%
\hspace{1cm}\medskip \mathbf{\Rightarrow }\textrm{ W }%
\begin{tabular}{l}
$\Gamma _{{}}\Rightarrow \Delta $ \\ \hline
$\Gamma _{{}}\Rightarrow \Delta ,\alpha $%
\end{tabular}%
\hspace{1cm}\medskip 
\textrm{Cut }%
\begin{tabular}{cc}
$\Gamma \Rightarrow \Delta ,\alpha $ & $\alpha ,\Gamma \Rightarrow \Delta $
\\ \hline
\multicolumn{2}{c}{$\Gamma \Rightarrow \Delta $}%
\end{tabular}%
\]%
\end{center}

\paragraph{Operational rules}

\begin{center}
\[
\lnot \Rightarrow \textrm{ }%
\begin{tabular}{l}
$\Gamma \Rightarrow \Delta ,\alpha $ \\ \hline
$\lnot \alpha ,\Gamma \Rightarrow \Delta $%
\end{tabular}%
\medskip \hspace{1cm}\Rightarrow \lnot \textrm{ }%
\begin{tabular}{l}
$\alpha ,\Gamma \Rightarrow \Delta $ \\ \hline
$\Gamma \Rightarrow \Delta ,\lnot \alpha $%
\end{tabular}%
\hspace{0.4cm} 
\]%
\[
\wedge \Rightarrow \textrm{ }%
\begin{tabular}{l}
$\alpha _{1},\alpha _{2},\Gamma \Rightarrow \Delta $ \\ \hline
$\alpha _{1}\wedge \alpha _{2},\Gamma \Rightarrow \Delta _{{}}$%
\end{tabular}%
\medskip \hspace{1cm}\Rightarrow \wedge \textrm{ }%
\begin{tabular}{cc}
$\Gamma \Rightarrow \Delta ,\alpha _{1}$ & $\Gamma \Rightarrow \Delta
,\alpha _{2}$ \\ \hline
\multicolumn{2}{c}{$\Gamma \Rightarrow \Delta ,\alpha _{1}\wedge \alpha _{2}$%
}%
\end{tabular}%
\hspace{0.4cm} 
\]%
\[
\vee \Rightarrow \textrm{ }%
\begin{tabular}{cc}
$\alpha _{1},\Gamma \Rightarrow \Delta $ & $\alpha _{2},\Gamma \Rightarrow
\Delta $ \\ \hline
\multicolumn{2}{c}{$\alpha _{1}\vee \alpha _{2},\Gamma \Rightarrow \Delta $}%
\end{tabular}%
\noindent \hspace{1cm}\Rightarrow \vee \textrm{ }%
\begin{tabular}{l}
$\Gamma \Rightarrow \Delta ,\alpha _{1},\alpha _{2}$ \\ \hline
$\Gamma \Rightarrow \Delta ,\alpha _{1}\vee \alpha _{2}{}_{{}}$%
\end{tabular}%
\hspace{0.4cm} 
\]%
\[
\rightarrow \Rightarrow \textrm{ }%
\begin{tabular}{cc}
$\Gamma \Rightarrow \Delta ,\alpha _{1}$ & $\alpha _{2},\Gamma \Rightarrow
\Delta $ \\ \hline
\multicolumn{2}{c}{$\alpha _{1}\rightarrow \alpha _{2},\Gamma \Rightarrow
\Delta $}%
\end{tabular}%
\noindent \hspace{1cm}\Rightarrow \rightarrow \textrm{ }%
\begin{tabular}{l}
$\alpha _{1},\Gamma \Rightarrow \Delta ,\alpha _{2}$ \\ \hline
$\Gamma \Rightarrow \Delta ,\alpha _{1}\rightarrow \alpha _{2}{}_{{}}$%
\end{tabular}%
\noindent \hspace{1cm} 
\]
\end{center}

\end{definition}

\bigskip

For $i=0,1,2$, consider the usual classical valuations from $For_i$ over  the set  $V_{\mathbf{CPL}}=\left\{ 1,0\right\} $ of \emph{classical truth-values}, where 1 denotes the \textquotedblleft true\textquotedblright\
value and 0 denotes the \textquotedblleft false\textquotedblright\ value. Let  $\vDash_{\mathbf{CPL}}$ be the semantical consequence relation of \textbf{CPL} over $For_0$, that is: $\Gamma \vDash_{\mathbf{CPL}} \alpha$ iff, for every classical valuation $v$: if $v(\gamma) =1$ for every $\gamma \in \Gamma$ then $v(\alpha) =1$. The following theorems are well known:

\begin{theorem} [Soundness and Completeness of  \textbf{C}] \label{compleCPL}
Let $\Gamma \cup\Delta$ be a finite set of formulas in $For_0$. Then: the sequent $\Gamma\Rightarrow\Delta$ is provable in \textbf{C} iff $\Gamma \vDash_{\mathbf{CPL}} \bigvee_{\alpha \in \Delta} \alpha$.\footnote{Here, $\bigvee_{\alpha \in \Delta} \alpha$ denotes the formula $\alpha_1 \vee (\alpha_2 \vee \cdots (\alpha_{n-1} \vee \alpha_n)\ldots)$, if $\Delta=\{\alpha_1,\ldots,\alpha_n\}$. If $\Delta=\{\alpha\}$ or $\Delta=\emptyset$ then $\bigvee_{\alpha \in \Delta} \alpha=\alpha$ and $\bigvee_{\alpha \in \Delta} \alpha=p_1\wedge\neg p_1$, respectively, where $p_1$ is the first propositional variable.} 
In particular: the sequent $\Gamma\Rightarrow\alpha$ is provable in \textbf{C} iff $\Gamma \vDash_{\mathbf{CPL}} \alpha$. The same holds for the  $\{\lnot,\vee\}$ and the $\{\lnot,\wedge\}$-fragments of \textbf{C}.
\end{theorem}

\begin{theorem} [Cut elimination for \textbf{C}] \label{cut-freeC}  Let $\Gamma \cup \Delta $ be a finite nonempty set of formulas in $For_0$.
If the sequent $\Gamma \Rightarrow \Delta$ is provable in \textbf{C} then there is a cut-free derivation of it in  \textbf{C}, that is,  a derivation without using the Cut rule. The same holds for the $\{\lnot,\vee\}$ and the $\{\lnot,\wedge\}$-fragments of \textbf{C}.
\end{theorem}

\section{Three-valued nonsense logics \textbf{B}$_{3}$ and \textbf{H}$%
_{3}$}

The logics of nonsense \textbf{B}$_{3}$ of Bochvar and \textbf{H}$_{3}$ of Halld\'{e}n 
are three-valued logics. Their set of truth-values is $V=\left\{ 1,{\frac12},0\right\} $ where the third non-classical truth-value $\frac12$ is interpreted as a nonsensical truth-value. In \textbf{H}$_{3}$ this third
truth-value is designated; on the other hand, $\frac12$ is undesignated in \textbf{B}$_{3}$. So, $D_{%
\mathbf{B}_3}=\left\{ 1\right\} $ is the set of designated values of \textbf{B}$_{3}$ and $D_{\mathbf{H}_3}=\left\{ 1,{\frac12}\right\} $ is the set of designated values of Halld\'en's logic \textbf{H}$%
_{3} $. The logic \textbf{B}$_{3}$ is defined over the signature $\Sigma_2^{\#_{\mathbf{B}}}$ 
obtained from  the signature $\Sigma_2$ by adding an unary `meaningful' connective $\#_{\mathbf{B}}$. By its turn, \textbf{H}$_{3} $ is defined over the signature $\Sigma_1^{\#_{\mathbf{H}}}$ 
obtained from $\Sigma_1$ by adding an unary `meaningful' connective $\#_{\mathbf{H}}$.
By means of the connectives $\#_{\mathbf{B}}$ and $\#_{\mathbf{H}}$ it is possible to express the meaninglessness of a formula at the object-language level of each logic. The abbreviations for defining the other classical connectives in each signature are the same as in \textbf{CPL} (recall Section~\ref{prelim}).
The truth-tables for negation, conjunction, disjunction, implication and  meaningful connectives in 
\textbf{B}$_{3}$ and \textbf{H}$_{3}$ are as follows:%
\[
\begin{tabular}{c|c}
& $\lnot $ \\ \hline
1 & 0 \\ 
$\frac12$
& 
$\frac12$
\\ 
0 & 1%
\end{tabular}%
\hspace{1cm}%
\begin{tabular}{c|ccc}
$\wedge $ & $1$ &
$\frac12$
& 0 \\ \hline
1 & 1 & 
$\frac12$
& 0 \\ 
$\frac12$
& 
$\frac12$
& 
$\frac12$
& 
$\frac12$
\\ 
0 & 0 & 
$\frac12$
& 0%
\end{tabular}
\hspace{1cm}%
\begin{tabular}{c|ccc}
$\vee $ & $1$ &
$\frac12$
& 0 \\ \hline
1 & 1 & 
$\frac12$
& 1 \\ 
$\frac12$
& 
$\frac12$
& 
$\frac12$
& 
$\frac12$
\\ 
0 & 1 & 
$\frac12$
& 0%
\end{tabular}
\hspace{1cm}%
\begin{tabular}{c|ccc}
$\to $ & $1$ &
$\frac12$
& 0 \\ \hline
1 & 1 & 
$\frac12$
& 0 \\ 
$\frac12$
& 
$\frac12$
& 
$\frac12$
& 
$\frac12$
\\ 
0 & 1 & 
$\frac12$
& 1%
\end{tabular}%
\]
\[
\begin{tabular}{c|c}
& $\#_{\mathbf{B}} $ \\ \hline
1 & 1 \\ 
$\frac12$
& 
$0$
\\ 
0 & 0%
\end{tabular}%
\hspace{1cm}%
\begin{tabular}{c|c}
& $\#_{\mathbf{H}} $ \\ \hline
1 & 1 \\
$\frac12$
& 
$0$
\\ 
0 & 1%
\end{tabular}%
\]

Additionally, $\#_{\mathbf{H}}$ can be defined in terms of the connectives of \textbf{B}$_{3}$ but the same relationship between $\#_{\mathbf{B}}$ and  \textbf{H}$_{3}$ is not true, and so the expressive power of  the matrices of  \textbf{B}$_{3}$ is strictly stronger than that of the matrices  of \textbf{H}$_{3}$ (cf.~\cite{cor:12}).

The key feature of both logics is the following, which can be easily proved by induction on the complexity of the formula $\alpha$:

\begin{proposition} \label{infec}
\label{infecciosidad} Let $\alpha$ be a formula of  \textbf{L} without $\#$ and let $v$ be a valuation
of  \textbf{L}, where \textbf{L} = \textbf{B}$_{3}$ or \textbf{L} = \textbf{H}$_{3}$.
Then: $v\left(\alpha\right) =
{\frac12}$ \ iff \ $v\left(p \right) =
{\frac12}$ for some propositional variable $p\in var\left( \alpha \right) $.
\end{proposition}

This means that in the `classical' fragment of \textbf{B}$_{3}$ and \textbf{H}$_{3}$ the non-classical truth-value $\frac12$
is  `infectious': an atomic formula `infects' complex formulas with the nonsensical
truth-value. It is easy to prove that, over the respective $\Sigma_i$,
both \textbf{B}$_{3}$ and \textbf{H}$_{3}$ are deductive fragments of
classical logic: every valid inference in \textbf{B}$_{3}$ or in \textbf{H}$%
_{3}$ written in the classical signature $\Sigma_i$ is valid in \textbf{CPL}. In fact, the
following proposition (whose proof is immediate) holds in \textbf{B}$_{3}$ and \textbf{H}$_{3}$.

\begin{proposition} \label{classi}
\label{clasicidad} Let $\alpha$ be a formula of  \textbf{L} without $\#$, let $v_{\mathbf{CPL}}$ be a classical
valuation and let $v_{\mathbf{L}}$ be a valuation of \textbf{L},
where \textbf{L} = \textbf{B}$_{3}$ or \textbf{L} = \textbf{H}$_{3}$. If $v_{\mathbf{L}}\left(
p\right) = v_{\mathbf{CPL}}\left( p \right) $ for every propositional variable $p\in
var\left( \alpha \right) $ then $v_{\mathbf{L}}\left( \alpha \right) =v_{%
\mathbf{CPL}}\left( \alpha \right) $ (and so $v_{\mathbf{L}}\left( \alpha
\right) \in \left\{ 1,0\right\} $).
\end{proposition}

Despite these similarities,
there are important differences  between \textbf{B}$_{3}$ and \textbf{H}$_{3}$ with respect to classical logic as a consequence of choosing different sets of designed truth-valued:

\begin{itemize}
\item There are no tautological formulas over $\Sigma_2$ in 
\textbf{B}$_{3}$; \textbf{H}$_{3}$ contain every classical tautology over $\Sigma_1$.

\item No contradiction written over $\Sigma_1$ is a trivializing formula in \textbf{H%
}$_{3}$; every contradiction over $\Sigma_2$ is a trivializing formula in \textbf{B}%
$_{3}$.

\item The Deduction Theorem is not valid in \textbf{B}$_{3}$ and \textit{%
modus ponens }is not valid in \textbf{H}$_{3}$. So, the following
metaproperty  does not hold in 
\textbf{B}$_{3}$: \textit{if }$\Gamma ,\alpha \nolinebreak \vDash \nolinebreak \beta 
$, \textit{then} $\Gamma \vDash \alpha \rightarrow \beta $; on the other hand, the following metaproperty  does not hold in \textbf{H}$_{3}$: \textit{if} $\Gamma \nolinebreak \vDash
\nolinebreak \nolinebreak \alpha \rightarrow \nolinebreak \beta $, \textit{%
then} $\Gamma ,\alpha \vDash \beta $.

\item The inference $\alpha \vDash \alpha \vee \beta $ does not hold in \textbf{B}$_{3}$; in \textbf{H}$_{3}$
the inference $\alpha \wedge \beta \vDash \alpha $ does not hold.

\item In \textbf{B}$_{3}$ the Principle of Excluded Middle: 
\begin{center}
\hspace*{4cm}$\vDash \alpha \vee \lnot \alpha  \hspace{4.5cm} \textrm{(PEM)}$
\end{center}
\noindent does not hold; in \textbf{H}$_{3}$ the Principle of Explosion: 
\begin{center}
\hspace*{4cm}$\alpha ,\lnot \alpha \vDash \beta  \hspace{4.5cm} \textrm{(PE)}$
\end{center}
\noindent does not hold. Thus,  \textbf{B}$_{3}$ is  paracomplete w.r.t.  the negation $\lnot$, while  \textbf{H}$_{3}$ is paraconsistent w.r.t. $\lnot$.
\end{itemize}

These differences between Bochvar and Halld\'en's connectives with respect
to classical connectives are not independent from each other, and their
connections are expressed in the following theorems, 
which constitute the basis of our proposal.

\begin{theorem}
\label{validezB} Let $\Gamma \cup \{\alpha \}$ be a set of formulas in $For_2$ such that $\Gamma \vDash _{\mathbf{CPL}}\alpha $. Then: 
\[
 \textrm{ if }var(\alpha )\subseteq
var(\Gamma )\textrm{ or }\Gamma \vDash _{\mathbf{CPL}} p_1 \wedge \lnot p_1
 \textrm{ then }\Gamma \vDash _{\mathbf{B}_3}\alpha
\textrm{.} 
\]
\end{theorem}

\begin{proof}
Cf. \cite{boch:38,will:94, mal:07, cor:12}.
\end{proof}

\begin{theorem}
\label{validezH} Let $\Gamma \cup \{\alpha \}$ be a set of formulas in $For_1$ such that 
$\Gamma \vDash _{\mathbf{CPL}}\alpha $. Then: 
\[
\textrm{ if }var(\Gamma )\subseteq
var(\alpha )\textrm{ or }\vDash _{\mathbf{CPL}}\alpha 
 \textrm{ then }
\Gamma \vDash _{\mathbf{H}_3}\alpha 
\textrm{.} 
\]
\end{theorem}

\begin{proof}
Assume that $\Gamma \vDash _{\mathbf{CPL}}\alpha $.  If  $\Gamma \nvDash _{\mathbf{H}_3}\alpha $
then there is a valuation $v_{\mathbf{H}_3}$  for $\mathbf{H}_3$ such that  
$v_{\mathbf{H}_3}\left( \Gamma \right) \subseteq \left\{ 1, {\frac12} \right\} $ and $v_{\mathbf{H}_3}\left( \alpha \right) =0$. Suppose that 
$var\left( \Gamma \right) \subseteq var\left( \alpha \right) $. Since 
$v_{\mathbf{H}_3}\left( \alpha \right) =0$ then, by Proposition~\ref{infec}, $v_{\mathbf{H}_3}\left( p\right) \in \left\{ 1,0\right\} $
for every propositional variable $p\in var\left( \alpha \right) $. Thus $v_{\mathbf{H}_3}\left( \Gamma \right) \subseteq \left\{ 1\right\} $. Let $v_{\mathbf{CPL}}$ be a classical valuation such that 
$v_{\mathbf{CPL}}\left( p\right) =v_{\mathbf{H}_3}\left( p\right) $ for every $p\in var\left( \alpha \right) =var\left( \alpha \right) \cup var\left(
\Gamma \right) $. Then, by Proposition~\ref{classi}, $v_{\mathbf{CPL}}\left(
\Gamma \right) \subseteq \left\{ 1\right\} $ but $v_{\mathbf{CPL}}\left(
\alpha \right) =0$, a contradiction. Then, $var\left( \Gamma \right) \nsubseteq var\left( \alpha
\right) $. Thus, if  $var\left( \Gamma \right) \subseteq var\left( \alpha
\right) $ then $\Gamma \vDash _{\mathbf{H}_3}\alpha $.

Finally,  if $\alpha$ is a classical tautology, let $v_{\mathbf{H}_3}$ be a valuation  for $\mathbf{H}_3$. If  $v_{\mathbf{H}_3}\left( p \right)={\frac12}$ for some $p \in var\left( \alpha \right) $ then  $v_{\mathbf{H}_3}\left( \alpha \right)={\frac12}$, by  Proposition~\ref{infec}. On the other hand, if 
$v_{\mathbf{H}_3}\left(var\left( \alpha \right) \right) \subseteq\{0,1\}$ then, by Proposition~\ref{classi},
$v_{\mathbf{H}_3}\left( \alpha \right)=1$. Then, $\vDash _{\mathbf{H}_3}\alpha$ and so $\Gamma \vDash _{\mathbf{H}_3}\alpha $.

\end{proof}

\bigskip So, by Theorem~\ref{validezB}, we have that if a valid classical
inference $\Gamma \vDash \alpha $ is invalid in Bochvar's nonsense
logic then $\Gamma $ is a consistent set of formulas of \textbf{CPL} such that $
var(\alpha )\varsubsetneq var(\Gamma )$. On the other hand, Theorem~\ref{validezH}
expresses that if a valid classical inference $\Gamma \vDash \alpha $ is invalid in
Halld\'en's nonsense logic then $\alpha $ is not a tautological formula in 
\textbf{CPL} and $var(\Gamma )\varsubsetneq var(\alpha )$. Therefore, it is
clear that $\vDash _{\mathbf{H}_3}\alpha $ but $\nvDash _{\mathbf{B}_3}\alpha$, for every $\alpha $ such that $\vDash_{\mathbf{CPL}}\alpha $.

By Theorems \ref{validezB} and \ref{validezH} we obtain a sufficient condition in order to determine whether a valid classical inference is also valid in both \textbf{B}$_{3}$ and 
\textbf{H}$_{3}$.

\begin{corollary}
Let $\Gamma \cup \{\alpha \}$ be a set of formulas in $For_0$ such that $\Gamma \vDash _{\mathbf{CPL}}\alpha $. Then: \footnote{Obviously we are identifying here a primitive connective of $\Sigma_0$ with its abbreviation in $\Sigma_i$, for $i=1,2$.} 
\[
\textrm{if }var\left( \Gamma \right) =var\left( \alpha \right) \textrm{, then }
\Gamma \vDash _{\mathbf{B}_3}\alpha \textrm{ and }\Gamma \vDash _{\mathbf{H}_3}
\alpha \textrm{.} 
\]
\end{corollary}

\bigskip We will introduce cut-free sequent calculi for the  $\{\lnot,\vee\}$-fragment of \textbf{H}$_{3}$ and for the $\{\lnot,\wedge\}$-fragment of \textbf{B}$_{3}$, where $\wedge$ and $\to$ ($\vee$ and $\to$, respectively) are derived connectives. 
The strategy adopted is to modify the classical sequent rules for
classical connectives by adding suitable provisos. As we shall see, the provisos are applied to symmetrical rules: in the fragment of  Halld\'en's logic, the provisos apply  to the introduction rules for conjunction, implication and negation on the left side of the sequent while, in the case of Bochvar's logic, 
the proviso applies  to the introduction rules for disjunction, implication and negation on the  the right side. This reflects the relationship between  these logics and classical logic, as depicted in theorems~\ref{validezB} and~\ref{validezH}.

\bigskip

\section{Sequent calculus \textbf{H} for the $\{\lnot,\vee\}$-fragment of Halld\'en's logic \textbf{H}$_{3}$}

As suggested by Theorem~\ref{validezH},
certain proofs in \textbf{C} should be blocked in any sequent calculus for \textbf{H}$_{3}$.
We present now a cut-free sequent calculus \textbf{H} for the fragment of \textbf{H}$_{3}$ over $\Sigma_1$ by adding provisos on the application of  (classical) rules such that the
construction of complex formulas in the antecedent of the sequents is
blocked in some cases. By symmetry, a sequent calculus \textbf{B} for \textbf{B}$_{3}$ will be also introduced by adding
provisos on the application of (classical) rules such that the
construction of complex formulas in the succedents of the sequents is
blocked under certain circumstances.  

\begin{definition}
The sequent calculus \textbf{H} is obtained from the $\{\lnot,\vee\}$-fragment of
 \textbf{C}  by replacing the rule $\lnot \Rightarrow$ by the following one:

\begin{center}
\[
\lnot ^{H}\Rightarrow \textrm{ }%
\begin{tabular}{l}
$\Gamma \Rightarrow \Delta ,\alpha $ \\ \hline
$\lnot \alpha ,\Gamma \Rightarrow \Delta $%
\end{tabular}%
\hspace{0.5cm} \mbox{provided that } \
var(\alpha) \subseteq var(\Delta)
\]
\end{center}
\end{definition}

\

\begin{proposition}
The following  rules are derivable in \textbf{H}: 
\[
\wedge ^{H}\Rightarrow \textrm{ }%
\begin{tabular}{l}
$\alpha _{1},\alpha _{2},\Gamma \Rightarrow \Delta $ \\ \hline
$\alpha _{1}\wedge \alpha _{2},\Gamma \Rightarrow \Delta _{{}}$%
\end{tabular}%
\medskip \hspace{1cm}\Rightarrow \wedge \textrm{ }%
\begin{tabular}{cc}
$\Gamma \Rightarrow \Delta ,\alpha _{1}$ & $\Gamma \Rightarrow \Delta
,\alpha _{2}$ \\ \hline
\multicolumn{2}{c}{$\Gamma \Rightarrow \Delta ,\alpha _{1}\wedge \alpha _{2}$%
}%
\end{tabular}%
\medskip 
\]
with the following proviso: $var(\{\alpha _{1},\alpha _{2}\}) \subseteq var(\Delta)$ in $\wedge ^{H}\mathbf{\Rightarrow }$.
\end{proposition}

\begin{proof}
Assume that $var(\{\alpha _{1},\alpha _{2}\}) \subseteq var(\Delta)$. Then 
$var(\lnot \alpha _{1}\vee \lnot\alpha _{2}) \subseteq var(\Delta)$ and so the following derivation can be done in  \textbf{H}: 
\[
\begin{tabular}{cc}
$\alpha _{1},\alpha _{2},\Gamma \Rightarrow \Delta$ & \\ \cline{1-1}
$\Gamma \Rightarrow \Delta ,\lnot
\alpha _{1},\lnot \alpha _{2} $ & (by $\Rightarrow  \lnot$) \\[2mm] 
\cline{1-1}
$\Gamma \Rightarrow \Delta ,\lnot
\alpha _{1}\vee \lnot \alpha _{2}$ & (by $\Rightarrow \vee$) \\[2mm] \cline{1-1}
$\lnot \left( \lnot \alpha _{1}\vee \lnot
\alpha _{2}\right),\Gamma \Rightarrow \Delta $ & (by $\lnot
^{H}\Rightarrow $)%
\end{tabular}%
\]

In order to obtain $\Rightarrow \wedge$, the following derivation can be done  in  \textbf{H}:%
\[
\displaystyle \frac{
 \displaystyle \frac{ \Gamma \Rightarrow \Delta,\alpha _{1}}{\displaystyle \frac{\Gamma \Rightarrow \Delta ,\lnot \left( \lnot \alpha _{1}\vee \lnot
\alpha _{2}\right), \alpha _{1}}{\lnot \alpha _{1},\Gamma \Rightarrow \Delta ,\lnot \left( \lnot \alpha _{1}\vee \lnot
\alpha _{2}\right)} \ \lnot ^{H}\Rightarrow} \Rightarrow W \hspace{1cm} 
\displaystyle \frac{\Gamma \Rightarrow \Delta, \alpha _{2}}{\displaystyle \frac{\Gamma \Rightarrow \Delta ,\lnot \left( \lnot \alpha _{1}\vee \lnot
\alpha _{2}\right), \alpha _{2}}{\lnot \alpha _{2},\Gamma \Rightarrow \Delta ,\lnot \left( \lnot \alpha _{1}\vee \lnot
\alpha _{2}\right)} \ \lnot ^{H}\Rightarrow} \Rightarrow W }
{\displaystyle \frac{\lnot \alpha _{1}\vee\lnot \alpha _{2},\Gamma \Rightarrow \Delta ,\lnot \left( \lnot \alpha _{1}\vee \lnot
\alpha _{2}\right)}{\Gamma \Rightarrow \Delta ,\lnot \left( \lnot \alpha _{1}\vee \lnot
\alpha _{2}\right)} \ \Rightarrow\lnot}\vee\Rightarrow 
\]%
\end{proof}

\begin{proposition}
The following  implicational rules are derivable in \textbf{H}: 
\[
\rightarrow ^{H}\mathbf{\Rightarrow }\textrm{ }%
\begin{tabular}{cc}
$\Gamma \Rightarrow \Delta ,\alpha _{1}$ & $\alpha _{2},\Gamma \Rightarrow
\Delta $ \\ \hline
\multicolumn{2}{c}{$\alpha _{1}\rightarrow \alpha _{2},\Gamma \Rightarrow
\Delta $}%
\end{tabular}%
\hspace{1cm}\mathbf{\Rightarrow }\rightarrow \textrm{ }%
\begin{tabular}{l}
$\alpha _{1},\Gamma \Rightarrow \Delta ,\alpha _{2}$ \\ \hline
$\Gamma \Rightarrow \Delta ,\alpha _{1}\rightarrow \alpha _{2}$%
\end{tabular}%
\]
with the following proviso: $var(\{\alpha _{1},\alpha _{2}\}) \subseteq var(\Delta)$ in $\rightarrow ^{H}\mathbf{\Rightarrow }$.
\end{proposition}

\begin{proof}
Straightforward, by considering that $\alpha _{1}\rightarrow \alpha _{2}$ stands for $\lnot\alpha _{1}\vee \alpha _{2}$ in \textbf{H}.
\end{proof}

\subsection{Soundness of \textbf{H}}

In this subsection we shall prove the soundness of sequent calculus \textbf{H%
}. Firstly,  some semantical notions will be extended from formulas to
sequents.

\begin{definition}
\label{modelosecuente} Let ${\mathbf{L}}$ be a matrix logic over a signature $\Sigma$. A valuation $v$ of  \textbf{L } is a 
\emph{model of a sequent} $\Gamma \Rightarrow \Delta$ over $\Sigma$ iff, if $v\left( \Gamma \right) \subseteq D_{\mathbf{L}}$, then $v\left(
\delta \right) \in D_{\mathbf{L}}$ for some $\delta \in \Delta $. When $v$ is a model of the sequent $\Gamma \Rightarrow \Delta $, we will
write $v\vDash_{\mathbf{L}} \Gamma \Rightarrow \Delta $.
\end{definition}

\begin{definition}
\label{secuentevalido} A \emph{sequent} $\Gamma \Rightarrow
\Delta $ is \emph{valid} in  \textbf{L} if, for every valuation $v$ of \textbf{L}, $v$ is a model of the sequent $\Gamma \Rightarrow \Delta $. When the sequent is valid, we will write $\vDash _{\mathbf{L}}\Gamma \Rightarrow \Delta $.
\end{definition}

It is worth noting that $\vDash _{\mathbf{L}}\Gamma \Rightarrow \alpha $ \ iff \  $\Gamma\vDash _{\mathbf{L}} \alpha $. Additionally,  $\vDash _{\mathbf{CPL}}\Gamma \Rightarrow \Delta $ iff $\Gamma \vDash_{\mathbf{CPL}} \bigvee_{\alpha \in \Delta} \alpha$

\begin{definition}
A sequent rule $\mathfrak{R}$ preserves \emph{validity} in \textbf{L} if, for every instance 
$\displaystyle\frac{\Upsilon}{S}$ of $\mathfrak{R}$ and for every  valuation $v$ of \textbf{L}, if  $v\vDash_{\mathbf{L}} S'$ for every $S' \in \Upsilon$ then $v\vDash_{\mathbf{L}} S$.
\end{definition}

\begin{lemma}
\label{reglasHvalidas} Every sequent rule of the calculus \textbf{H} preserves
validity.
\end{lemma}
\begin{proof}
Observe that the axiom Ax and the structural rules preserve validity, since they correspond to properties which are valid in every Tarskian logic (and  \textbf{H} is Tarskian since it is a matrix logic).   

\begin{description}

\item[$\Rightarrow \lnot $] Let $v$ be a valuation of \textbf{H}$_{3}$ such that $v\vDash _{\mathbf{H}_3}\alpha ,\Gamma
\Rightarrow \Delta $, and suppose that $v\left( \Gamma \right) \subseteq \left\{ 1,
{\frac12}\right\} $. If $v\left( \lnot \alpha \right) =0$, then $v\left( \alpha \right) =1$. Then, by hypothesis, we infer that $v\left( \delta\right) \in \left\{ 1,
{\frac12}\right\} $, for some $\delta\in \Delta $. If $v\left(
\lnot \alpha \right) \neq 0$, then $v\left( \lnot \alpha
\right) \in \left\{ 1, {\frac12}\right\} $. This shows that $v\vDash _{\mathbf{H}_3}\Gamma \Rightarrow \Delta,\lnot \alpha$.

\item[$\lnot ^{H}\Rightarrow $] Let $v$ be a valuation of \textbf{H}$_{3}$ such that 
$v\vDash_{\mathbf{H}_3}\Gamma
\Rightarrow \Delta ,\alpha $ and assume that $var\left( \alpha \right)
\subseteq var\left( \Delta \right) $. Suppose that $v\left( \lnot \alpha \right) \in \left\{ 1,{\frac12}\right\} $ and $v\left( \Gamma \right) \subseteq \left\{ 1,{\frac12}\right\}$. Then, by hypothesis, $v\left( \delta\right)
\in \left\{ 1,{\frac12}
\right\} $, for some $\delta\in \Delta $, or $v\left(
\alpha \right) \in \left\{ 1,{\frac12}\right\} $. Since $v\left( \lnot \alpha \right) \in \left\{ 1,{\frac12}\right\} $, then $v\left( \alpha \right) \in \left\{ 0,{\frac12}\right\} $. If $v\left( \alpha \right) =0$ then $v\left( \delta\right) \in \left\{ 1,{\frac12}\right\} $, for some $\delta\in \Delta $. And if $v\left(
\alpha \right) ={\frac12}$, then, by Proposition \ref{infecciosidad}, we infer that $v\left( p\right) =
{\frac12}$ for some atomic formula $p\in var\left( \alpha \right) $. Since $var\left( \alpha \right) \subseteq var\left( \Delta\right) $ then $p \in var(\delta)$ for some $\delta\in \Delta $ and so, again by Proposition \ref{infecciosidad}, we infer that $v\left( \delta\right) =\left\{ {\frac12}\right\} $. Therefore, we conclude that $v\vDash _{\mathbf{H}_3}\lnot \alpha ,\Gamma \Rightarrow \Delta $.

\item[$\Rightarrow \vee $] Let $v$ be a valuation of \textbf{H}$_{3}$ such that $v\vDash _{\mathbf{H}_3} \Gamma \Rightarrow \Delta, \alpha
_{1},\alpha _{2} $ and assume that $v\left( \Gamma \right) \subseteq \left\{ 1,
{\frac12}\right\} $. If $v\left( \alpha _{1}\right)= v\left( \alpha _{2}\right) =0$ 
then, by hypothesis, we infer that $v\left( \delta
\right) \in \left\{ 1,{\frac12}\right\} $, for some $\delta \in \Delta $. Therefore  $v\vDash _{\mathbf{H}_3} \Gamma \Rightarrow \Delta, \alpha
_{1}\vee \alpha _{2}$. Otherwise, if $v\left( \alpha _{1}\right) \in \left\{ 1,
{\frac12}\right\}$ or  $v\left( \alpha _{2}\right) \in \left\{ 1,
{\frac12}\right\}$ then $v\left( \alpha
_{1}\vee \alpha _{2}\right) \in \left\{ 1,
{\frac12}\right\}$ and so  $v\vDash _{\mathbf{H}_3} \Gamma \Rightarrow \Delta, \alpha
_{1}\vee \alpha _{2}$.

\item[$\vee\Rightarrow $] Let $v$ be a valuation of \textbf{H}$_{3}$ such that $v\vDash _{\mathbf{H}_3}\alpha _{1}, \Gamma \Rightarrow
\Delta$ and $v\vDash _{\mathbf{H}_3}\alpha _{2}, \Gamma \Rightarrow \Delta$. Suppose
that $v\left(  \alpha_{1}\vee \alpha _{2}\right) \in \left\{ 1,
{\frac12}\right\}$ and 
$v\left( \Gamma \right) \subseteq \left\{ 1,{\frac12}\right\} $. Then, either  $v\left( \alpha _{1}\right) \in \left\{ 1,
{\frac12}\right\}$ or  $v\left( \alpha _{2}\right) \in \left\{ 1,
{\frac12}\right\}$.
By hypothesis, it follows that $v\left( \delta
\right) \in \left\{ 1,{\frac12}\right\} $, for some $\delta \in \Delta $ and so  $v\vDash _{\mathbf{H}_3} \alpha
_{1}\vee \alpha _{2}, \Gamma \Rightarrow \Delta$.
\end{description}
\end{proof}

\begin{theorem} [Soundness of  \textbf{H}]
\label{correccionseqH} Let $\Gamma \cup \Delta $ be a set of formulas in $For_1$. Then: if $\Gamma \Rightarrow \Delta$ is provable in \textbf{H}  then $\vDash _{\mathbf{H}_3}\Gamma \Rightarrow \Delta$. In particular, if $\Gamma \Rightarrow \alpha$ 
is provable in \textbf{H}  then $\Gamma \vDash _{\mathbf{H}_3}\alpha$.
\end{theorem}
\begin{proof}
If the sequent $\Gamma \Rightarrow \Delta $ is an instance of axiom Ax, then $\Gamma \Rightarrow \Delta$ is valid in \textbf{H}$_3$. By induction on the depth of a derivation of
 $\Gamma \Rightarrow \Delta $ in  \textbf{H} it follows, by the previous Lemma~\ref{reglasHvalidas}, that  the sequent $\Gamma \Rightarrow \Delta $ is valid in \textbf{H}$_3$.
\end{proof}

\begin{proposition}
Let $\Gamma $ be a nonempty set of formulas in $For_1$. Then the sequent $\Gamma \Rightarrow $ \ 
is not provable in~\textbf{H}.
\end{proposition}

\begin{proof}
 Let $v$ be a 
\textbf{H}$_{3}$-valuation such that $v\left( p\right) ={\frac12}$ for every $p\in var\left( \Gamma \right) $. Then
$v\nvDash _{\mathbf{H}_3}\Gamma \Rightarrow$ and so 
$\nvDash _{\mathbf{H}_3}\Gamma \Rightarrow$~. By contraposition of
Theorem \ref{correccionseqH}, we conclude that the sequent $\Gamma \Rightarrow $ \  
is not provable in \textbf{H}.
\end{proof}

\subsection{Completeness of \textbf{H}}

The following result follows straightforwardly:

\begin{proposition}
\label{CextensionH} Let $\Gamma \cup \Delta $ be a  finite nonempty set of formulas in $For_1$. Then: 
\[
\textrm{ if }\vDash _{\mathbf{H}_3}\Gamma \Rightarrow \Delta \textrm{, then }%
\vDash _{\mathbf{CPL}}\Gamma \Rightarrow \Delta \textrm{.} 
\]
\end{proposition}
\begin{proof} Assume that $\vDash _{\mathbf{H}_3}\Gamma \Rightarrow \Delta$ and let $v$ be a classical valuation such that $v\left(\Gamma\right) \subseteq \{1\}$. By Proposition~\ref{classi}, $v$ can be seen as a \textbf{H}$_{3}$-valuation such that   $v\left(\Gamma\right) \subseteq \{1,{\frac12}\}$. By
hypothesis, $v\left(\delta\right) \in  \{1,{\frac12}\}$ for some $\delta \in \Delta$. Since $v$ is classical, it follows that $v\left(\delta\right)=1$ for some $\delta \in \Delta$, therefore $\vDash _{\mathbf{CPL}}\Gamma \Rightarrow \Delta$.
\end{proof}

\begin{proposition}
\label{HextensionC} Let $\Gamma \cup \Delta $ be a finite nonempty set of formulas in $For_1$. Then: 
if $\Gamma \Rightarrow \Delta$ is provable in \textbf{H}  then $\Gamma \Rightarrow \Delta$ is provable 
in the $\{\lnot,\vee\}$-fragment of  \textbf{C}. 
\end{proposition}

\begin{proof}
This is obvious, since  \textbf{H} is a restricted version of the $\{\lnot,\vee\}$-fragment of  \textbf{C}.
\end{proof}

\begin{lemma}
\label{variablesincluidasH} Let $\Gamma \cup \Delta $ be a finite nonempty set of formulas in 
$For_1$. Then: 
if $\Gamma \Rightarrow \Delta$ is provable 
in the $\{\lnot,\vee\}$-fragment of \textbf{C}  and $var\left( \Gamma \right) \subseteq var\left( \Delta \right)$  then $\Gamma \Rightarrow \Delta$ is provable in \textbf{H} without using the Cut rule.
\end{lemma}

\begin{proof}
Recall that derivations in \textbf{C} and \textbf{H} are rooted
binary trees such that the root is the sequent being proved, and
the leaves are always instances of the axiom Ax of the form $\alpha
\Rightarrow \alpha $ for some formula $\alpha$.

Assume that $\Pi$ is a cut-free derivation in the $\{\lnot,\vee\}$-fragment of  \textbf{C} of a
sequent $\Gamma \Rightarrow \Delta$ such that $var\left( \Gamma
\right) \subseteq var\left( \Delta \right)$ (we can assume this
by Theorem~\ref{cut-freeC}). If $\Pi$ is also a derivation in
\textbf{H} then the result follows automatically. Otherwise,
there are in $\Pi$, by force, applications of the rule
$\lnot \Rightarrow$, namely
$$\lnot \Rightarrow \textrm{ }%
\begin{tabular}{l}
$\Gamma' \Rightarrow \Delta' ,\alpha $ \\ \hline
$\lnot \alpha ,\Gamma' \Rightarrow \Delta' $%
\end{tabular}$$
such that the proviso required by this rule in \textbf{H} is
not satisfied. Since $\Pi$ is cut-free then the set of variables occurring in the root sequent
$\Gamma \Rightarrow \Delta$ contains all the propositional
variables occurring in $\Pi$. Then, by hypothesis, all the
propositional variables occurring in $\Pi$ belong to the set
$var\left( \Delta \right)$. Consider now the derivation $\Pi'$ in
\textbf{C} obtained from $\Pi$ in two steps: firstly, the right-hand side of each sequent (that is, of each node) of $\Pi$ is enlarged  by adding simultaneously all the formulas in $\Delta$. This generates a  
rooted binary tree $\Pi_0$ whose leafs are sequents of the form  $\alpha \Rightarrow \alpha, \Delta $. Each of such leaves of $\Pi_0$ corresponds to the original  occurrence of an  axiom (that is, a leaf)  $\alpha \Rightarrow \alpha $ in the derivation $\Pi$.  In the second step, we replace each leaf $\alpha \Rightarrow \alpha, \Delta $ of  $\Pi_0$ by a branch started by $\alpha \Rightarrow \alpha$ and followed by iterated applications of the weakening rule $\Rightarrow W$ until obtaining the sequent $\alpha \Rightarrow \alpha, \Delta $.  The resulting
rooted binary tree $\Pi'$ is clearly a (cut-free) derivation in the $\{\lnot,\vee\}$-fragment of 
\textbf{C} of the sequent $\Gamma \Rightarrow \Delta$.\footnote{Observe that some applications of  the weakening rule $\Rightarrow W$ in $\Pi$ may be innocuous in $\Pi'$.} But the critical
applications of the rule $\lnot \Rightarrow$ mentioned above have in $\Pi'$ the form
$$\lnot \Rightarrow \textrm{ }%
\begin{tabular}{l}
$\Gamma' \Rightarrow \Delta' ,\Delta,\alpha $ \\ \hline
$\lnot \alpha ,\Gamma' \Rightarrow \Delta',\Delta $.
\end{tabular}$$
Being so, these applications are allowed in
\textbf{H} (since all the propositional variables occurring in
$\Pi'$ belong to the set $var\left( \Delta \right)$) and so
$\Pi'$ is in fact a cut-free derivation in \textbf{H} of  the sequent $\Gamma \Rightarrow
\Delta$. That is, $\Gamma \Rightarrow \Delta$ is
provable in \textbf{H} without using the Cut rule.
\end{proof}

\begin{corollary}
\label{tautologiasH} Let $\Delta $ be a finite nonempty set of formulas in $For_1$. Then:
$\ \Rightarrow \Delta$ is provable in the $\{\lnot,\vee\}$-fragment of  \textbf{C} if and only if $\ \Rightarrow \Delta$ is provable in \textbf{H}.
\end{corollary}

\begin{corollary}
[Modus Ponens] \label{modusponensH} Let $\alpha,\beta \in For_1$. 
Then: if $\ \Rightarrow \alpha$ and $\ \Rightarrow \alpha \rightarrow \beta$ 
are provable in \textbf{H}
 then $\ \Rightarrow \beta$ is provable in \textbf{H}. 
\end{corollary}

\begin{lemma}
\label{variablesnoincluidasH} Let $\Gamma \cup \Delta $ be a finite nonempty set of formulas in $For_1$. 
If $\vDash_{\mathbf{H}_3}\Gamma \Rightarrow \Delta $ but $var\left( \Gamma \right)
\varsubsetneq var\left( \Delta \right) $ then there exists $\Gamma'
\subset \Gamma $ such that $\vDash _{\mathbf{H}_3}\Gamma'\Rightarrow
\Delta $, where $var\left( \Gamma'\right) \subseteq var\left( \Delta
\right) $.
\end{lemma}

\begin{proof}
Observe that if $\vDash _{\mathbf{H}_3}\Gamma \Rightarrow
\Delta $ then $\Delta \neq \emptyset $.

Assume that $\vDash _{\mathbf{H}_3}\Gamma \Rightarrow \Delta $ such that $var\left( \Gamma \right) \varsubsetneq var\left( \Delta \right) $. So, given a valuation $v$ for $\mathbf{H}_3$, 
if $v\left( \Gamma \right) \subseteq \left\{ 1,
{\frac12}
\right\} $ then $v\left( \delta \right) \in \left\{ 1,
{\frac12}
\right\} $ for some formula $\delta \in \Delta $. Given that $var\left( \Gamma
\right) \varsubsetneq var\left( \Delta \right) $ consider the set $\Gamma'
=\Gamma \setminus\left\{ \gamma \in \Gamma \ : \ var\left( \gamma\right) \varsubsetneq var\left( \Delta \right) \right\} $.
Then, $\Gamma'\subset \Gamma $ and $var\left( \Gamma'\right)
\subseteq var\left( \Delta \right) $. Let $v$ be a valuation  for $\mathbf{H}_3$ such that  
$v\left( \Gamma \textrm{'}\right) \subseteq \left\{ 1,
{\frac12}
\right\} $. If ${\frac12} \in v\left( \Gamma \textrm{'}\right) $ then $v\left( p\right) =
{\frac12}
$ for some propositional variable $p\in var\left( \Gamma'\right) $.
Since $var\left( \Gamma'\right) \subseteq var\left( \Delta \right) $, then $
{\frac12}\in v\left( \Delta \right) $. If $v\left( \Gamma'\right) \subseteq \left\{ 1\right\} $, suppose that 
$v\left( \Delta \right) =\left\{ 0\right\} $. Then $v\left( p\right) \in \left\{ 1,0\right\} $ for
every propositional variable $p\in var\left( \Delta \right) $, by Proposition~\ref{infec}. Since $var\left( \Gamma'\right) \subseteq var\left( \Delta \right) $ then, for every propositional variable $p\in var\left( \Gamma'\right) $, $v\left(
p\right) \in \left\{ 1,0\right\} $. Consider now a valuation $v'$ for $\mathbf{H}_3$ such that  
$v'\left( p\right) ={\frac12}$ for every $p\in var\left( \Gamma\right) \setminus var\left( \Delta\right)$, and  
$v'\left(p\right)= v\left(p\right) $ for every $p\in var\left( \Delta \right)$. Then, $v'\left( \Gamma\right)  \subseteq \left\{ 1,{\frac12}\right\} $. But then, by hypothesis, $v'\left( \delta \right)
\in \left\{ 1,{\frac12}\right\} $, for some $\delta \in \Delta $. That is, $v\left( \delta \right)
\in \left\{ 1,{\frac12}\right\} $ for some $\delta \in \Delta $, a contradiction.
Therefore, if $v\left( \Gamma'\right) \subseteq \left\{
1\right\} $ then $v\left( \delta \right) \neq 0$, for some $\delta \in \Delta $. So, $\vDash _{\mathbf{H}_3}\Gamma $'$\Rightarrow \Delta $.
\end{proof}

\begin{theorem}
[Completeness  of \textbf{H}] \label{completudsecH} Let $\Gamma \cup \Delta $ be a finite nonempty set of formulas in $For_1$.
If $\vDash_{\mathbf{H}_3}\Gamma \Rightarrow \Delta$ then
$\Gamma \Rightarrow \Delta$ is provable in \textbf{H} without using the Cut rule. In particular, if $\Gamma \vDash _{\mathbf{H}_3}\alpha$
 then the sequent $\Gamma \Rightarrow \alpha$ is provable in \textbf{H}, for every finite set $\Gamma \cup \{\alpha\}$.
\end{theorem}
\begin{proof}
Assume that $\vDash _{\mathbf{H}_3}\Gamma \Rightarrow \Delta $. Then, by
Proposition \ref{CextensionH}, $\vDash _{\mathbf{CPL}}\Gamma \Rightarrow
\Delta $. By Theorem~\ref{compleCPL}, $\Gamma \Rightarrow \Delta$ is provable in the $\{\lnot,\vee\}$-fragment of   \textbf{C}. If $var\left( \Gamma \right) \subseteq
var\left( \Delta \right) $ then, by Lemma~\ref{variablesincluidasH}, $\Gamma \Rightarrow \Delta$ is provable in \textbf{H} without using the Cut rule. If $var\left( \Gamma \right)
\varsubsetneq var\left( \Delta \right) $ then, by Lemma \ref%
{variablesnoincluidasH}, there exist a set $\Gamma'\subset \Gamma $ such
that $\vDash _{\mathbf{H}_3}\Gamma'\Rightarrow \Delta $, where $var\left(
\Gamma'\right) \subseteq var\left( \Delta \right) $. Then, using
Proposition~\ref{CextensionH} and Theorem~\ref{compleCPL} again,
we obtain that $\Gamma' \Rightarrow \Delta $ is provable in the $\{\lnot,\vee\}$-fragment of  \textbf{C}.
Since $var\left( \Gamma'\right) \subseteq var\left( \Delta \right)$ then, by using Lemma~\ref{variablesincluidasH}, it follows that $\Gamma'\Rightarrow \Delta $  is provable in \textbf{H} without using the Cut rule.  By applying the structural rule $W\Rightarrow$ several times we obtain  a derivation of  $\Gamma \Rightarrow \Delta$ in \textbf{H} without using the Cut rule, as desired.
\end{proof}

\begin{corollary}
[Cut elimination for \textbf{H}] \label{cut-freeH} Let $\Gamma \cup \Delta $ be a finite nonempty set of formulas in $For$.
If the sequent $\Gamma \Rightarrow \Delta$ is provable in \textbf{H} then there is a cut-free derivation of it in  \textbf{H}.
\end{corollary}
\begin{proof}
Suppose that $\Gamma \Rightarrow \Delta$ is provable in \textbf{H}. By Theorem~\ref{correccionseqH}, $\vDash_{\mathbf{H}_3}\Gamma \Rightarrow \Delta$. Then, by Theorem~\ref{completudsecH}, there is a  cut-free  derivation of  $\Gamma \Rightarrow \Delta$ in \textbf{H}. 
\end{proof}

\section{Sequent calculus \textbf{B} for the $\{\lnot,\wedge\}$-fragment of  Bochvar's logic \textbf{B}$_{3}$}

In this section we introduce the sequent calculus  \textbf{B} which will result cut-free, sound and complete for the conjunction-negation fragment of the nonsense logic \textbf{B}$_{3}$, where $\vee$ and $\to$ are derived connectives. As we shall see, there exists a symmetry between the provisos imposed in the rules of \textbf{B} and those imposed in \textbf{H}, as long as the language $\neg$, $\wedge$, $\vee$, $\to$ is considered.

\begin{definition}
The sequent calculus \textbf{B} is obtained from the $\{\lnot,\wedge\}$-fragment of 
 \textbf{C}  by replacing the rule $\Rightarrow \lnot$ by the following one:

\begin{center}
\[\Rightarrow \lnot ^{B}%
\begin{tabular}{c}
$\alpha ,\Gamma \Rightarrow \Delta $ \\ \hline
$\Gamma \Rightarrow \Delta ,\lnot \alpha $%
\end{tabular}%
\hspace{0.5cm} \mbox{provided that } \
var(\alpha) \subseteq var(\Gamma)
\]

\end{center}

\end{definition}

\

\begin{proposition}
The following  rules are derivable in \textbf{B}: 
\[
\vee \Rightarrow \textrm{ }%
\begin{tabular}{cc}
$\alpha _{1},\Gamma \Rightarrow \Delta $ & $\alpha _{2},\Gamma \Rightarrow
\Delta $ \\ \hline
\multicolumn{2}{c}{$\alpha _{1}\vee \alpha _{2},\Gamma \Rightarrow \Delta $}%
\end{tabular}%
\noindent \hspace{1cm}\Rightarrow \vee ^{B}\textrm{ }%
\begin{tabular}{l}
$\Gamma \Rightarrow \Delta ,\alpha _{1},\alpha _{2}$ \\ \hline
$\Gamma \Rightarrow \Delta ,\alpha _{1}\vee \alpha _{2}{}_{{}}$%
\end{tabular}%
\medskip 
\]
with the following proviso: $var(\{\alpha _{1},\alpha _{2}\}) \subseteq var(\Gamma)$ in $\Rightarrow \vee ^{B}$.
\end{proposition}

\begin{proof}
We leave the easy proof as an exercise to the reader.
\end{proof}

\bigskip

\begin{proposition}
The following  implicational rules are derivable  in \textbf{B}:%
\[
\rightarrow \mathbf{\Rightarrow }\textrm{ }%
\begin{tabular}{cc}
$\Gamma \Rightarrow \Delta ,\alpha _{1}$ & $\alpha _{2},\Gamma \Rightarrow
\Delta $ \\ \hline
\multicolumn{2}{c}{$\alpha _{1}\rightarrow \alpha _{2},\Gamma \Rightarrow
\Delta $}%
\end{tabular}%
\hspace{1cm}\mathbf{\Rightarrow }\rightarrow ^{B}%
\begin{tabular}{l}
$\alpha _{1},\Gamma \Rightarrow \Delta ,\alpha _{2}$ \\ \hline
$\Gamma \Rightarrow \Delta ,\alpha _{1}\rightarrow \alpha _{2}$%
\end{tabular}%
\]
with the following proviso: $var(\{\alpha _{1},\alpha _{2}\}) \subseteq var(\Gamma)$ in $\mathbf{\Rightarrow }\rightarrow ^{B}$.
\end{proposition}

\begin{proof}
The proof is also left to the reader.
\end{proof}

\subsection{Soundness of \textbf{B}}

In order to prove the Soundness Theorem for \textbf{B}, we will prove that
every sequent rule of the calculus \textbf{B} preserves
validity.

\begin{lemma}
\label{reglasBvalidas} Every sequent rule of the calculus \textbf{B} preserves
validity.
\end{lemma}
\begin{proof} 
As in the case of \textbf{H}, it is enough to analyze the rules for connectives.

\begin{description}

\item[$\Rightarrow \lnot ^{B}$] Assume that $v\models _{\mathbf{B}_3}\alpha ,\Gamma
\Rightarrow \Delta $ for some valuation $v$ in $\mathbf{B}_3$, where $var\left( \alpha \right) \subseteq
var\left( \Gamma \right) $. Suppose that $v\left( \Gamma
\right) \subseteq \left\{ 1\right\} $. Then, by Proposition~\ref{infecciosidad}, $v\left( p\right) \in \left\{ 1,0\right\}$, for every propositional variable $p$ such that $p\in var\left( \Gamma
\right) $. Since $var\left( \alpha \right) \subseteq var\left( \Gamma
\right) $, then $v\left( p\right) \in \left\{ 1,0\right\}$, for every  propositional variable $p\in var\left( \alpha \right) $. By Proposition~\ref{infecciosidad} again, we obtain that $v\left(
\alpha \right) \in \left\{ 1,0\right\} $. If $v\left( \alpha
\right) =1$, then by hypothesis, we obtain that $\left\{ 1\right\} \subseteq
v\left( \Delta \right) $. If $v\left( \alpha
\right) =0$ then $v\left( \lnot \alpha \right) =1$.
In both cases it follows that $\left\{ 1\right\} \subseteq
v\left( \Delta \cup \left\{\lnot\alpha\right\} \right) $.  Therefore $v\models _{\mathbf{B}_3}\Gamma \Rightarrow
\Delta ,\lnot \alpha $.

\item[$\lnot \Rightarrow $] Assume that $v\models _{\mathbf{B}_3}\Gamma \Rightarrow
\Delta ,\alpha $ for some valuation $v$ in $\mathbf{B}_3$. Suppose
that $v\left( \lnot \alpha \right) =1$ and $v\left( \Gamma \right) \subseteq \left\{ 1\right\} $. So, $\left\{ 1\right\}
\subseteq v\left( \Delta \right) $ or $v\left(
\alpha \right) =1$, by hypothesis. But, since $v\left( \lnot \alpha \right)
=1$, then $v\left( \alpha \right) =0$. Thus, $\left\{ 1\right\}
\subseteq v\left( \Delta \right) $ and so $v\models _{\mathbf{B}_3}\lnot \alpha ,\Gamma \Rightarrow
\Delta $.

\item[$\Rightarrow \wedge$] Assume that $v\models _{\mathbf{B}_3}\Gamma
\Rightarrow \Delta ,\alpha _{1}$ and $v\models _{\mathbf{B}_3}\Gamma
\Rightarrow \Delta ,\alpha _{2}$  for some valuation $v$ in $\mathbf{B}_3$. Suppose that $v\left( \Gamma \right) \subseteq \left\{ 1\right\} $. By hypothesis, we
obtain that either $\left\{ 1\right\} \subseteq v\left( \Delta \right) 
$ or both $v\left( \alpha _{1}\right) =1$ and $v\left(
\alpha _{2}\right) =1$. In both cases it follows that $\left\{ 1\right\} \subseteq
v\left( \Delta \cup \left\{ \left( \alpha _{1}\wedge \alpha
_{2}\right) \right\} \right) $. Then $v\models _{\mathbf{B}_3}\Gamma
\Rightarrow \Delta ,\alpha _{1}\wedge \alpha _{2}$.

\item[$\wedge \Rightarrow $] Assume that $v\models _{\mathbf{B}_3}\alpha _{1},\alpha
_{2},\Gamma \Rightarrow \Delta $  for some valuation $v$ in $\mathbf{B}_3$. Suppose that 
$v\left(\alpha _{1}\wedge \alpha _{2}\right) =1$ and $v\left(\Gamma\right)
\subseteq \left\{ 1\right\} $. So, $v\left( \alpha _{1}\right) =v\left( \alpha _{2}\right) =1$ and $v\left( \Gamma \right)\subseteq \left\{ 1\right\} $. By hypothesis, $\left\{ 1\right\}
\subseteq v\left( \Delta \right) $. Therefore,  $v\models _{\mathbf{B}_3}\alpha _{1}\wedge \alpha _{2},\Gamma \Rightarrow \Delta $. 
\end{description}
\end{proof}

As a consequence of this it follows the soundness theorem for  \textbf{B}:

\begin{theorem}
[Soundness of  \textbf{B}] \label{correccionseqB} Let $\Gamma \cup \Delta$ be a finite nonempty subset of $For_2$. Then: if  $\Gamma \Rightarrow \Delta$ is provable in  \textbf{B} then $\models _{\mathbf{B}_3}\Gamma \Rightarrow \Delta$. In particular,  if  $\Gamma \Rightarrow \alpha$ is provable in  \textbf{B} then $\Gamma\models _{\mathbf{B}_3} \alpha$.
\end{theorem}

\begin{corollary}
Let $\Delta \subseteq For_2$ be a nonempty set of formulas. Then the sequent $\ \Rightarrow \Delta$ is not provable in   \textbf{B}.
\end{corollary}
\begin{proof} Consider a valuation $v$ for  $\mathbf{B}_3$ such that $v(p)={\frac12}$ for every $p \in var(\Delta)$. Then $v\not\models _{\mathbf{B}_3} \Rightarrow \Delta $ and so $\not\models _{\mathbf{B}_3}\Rightarrow \Delta$. By Theorem~\ref{correccionseqB},
the sequent $\ \Rightarrow \Delta$ is not provable in   \textbf{B}.
\end{proof}

\subsection{Completeness of \textbf{B}}

The proof of completeness of  \textbf{B} is similar to that of  \textbf{H} and so we will omit some proofs.

\begin{proposition}
\label{BextensionC} Let $\Gamma \cup \Delta$ be a finite nonempty subset of $For_2$. Then: 
\[
\textrm{if }\models _{\mathbf{B}_3}\Gamma \Rightarrow \Delta \textrm{, then }%
\models _{\mathbf{CPL}}\Gamma \Rightarrow \Delta \textrm{.} 
\]
\end{proposition}

\begin{proposition}
\label{CextensionB} Let $\Gamma \cup \Delta $ be a finite nonempty subset of $For_2$. Then: 
if $\Gamma \Rightarrow \Delta$ is provable in \textbf{B}  then it is provable in the $\{\lnot,\wedge\}$-fragment of  \textbf{C}. 
\end{proposition}

\begin{lemma}
\label{variablesincluidasB} Let $\Gamma \cup \Delta$  be a finite nonempty subset of $For_2$. Then: 
if $\Gamma \Rightarrow \Delta$   is provable in the $\{\lnot,\wedge\}$-fragment of  \textbf{C} and 
$var\left( \Delta \right) \subseteq var\left( \Gamma \right)$  then $\Gamma \Rightarrow \Delta$ is provable in \textbf{B} without using the Cut rule.
\end{lemma}
\begin{proof}
The proof is analogous to that of Lemma~\ref{variablesincluidasH}, but now using the rule $W\Rightarrow$.
\end{proof}

\begin{lemma}
\label{variablesnoincluidasB} Let $\Gamma \cup \Delta$   be a finite nonempty subset of $For_2$. 
If $\models _{\mathbf{B}_3}\Gamma \Rightarrow \Delta$ but $var\left( \Delta
\right) \varsubsetneq var\left( \Gamma \right) $ then there exist a set $\Delta'\subset \Delta $ such that  $\models _{\mathbf{B}_3}\Gamma \Rightarrow \Delta'$, where $var\left( \Delta'\right) \subseteq
var\left( \Gamma \right) $.
\end{lemma}
\begin{proof} Let $\Delta'
=\Delta \setminus\left\{ \delta \in \Delta \ : \ var\left( \delta\right) \varsubsetneq var\left( \Gamma \right) \right\}$. Suppose that there is a  $\mathbf{B}_3$-valuation $v$ such that $v\left(\Gamma \right)\subseteq\left\{1 \right\}$ but
$v\left(\Delta' \right)\subseteq\left\{0,{\frac12} \right\}$.  Thus, the  $\mathbf{B}_3$-valuation $v'$ such that $v'\left(p \right)=v\left(p \right)$ for every $p \in var\left( \Gamma \right)$ and $v'\left(p' \right)=
{\frac12}$ for every $p' \in var\left( \Delta\right) \setminus var\left( \Gamma \right)$ is such that
$v'\left(\Gamma \right)\subseteq\left\{1 \right\}$ but
$v'\left(\Delta \right)\subseteq\left\{0,{\frac12} \right\}$, a contradiction. Therefore $\models _{\mathbf{B}_3}\Gamma \Rightarrow \Delta'$, where $var\left( \Delta'\right) \subseteq
var\left( \Gamma \right) $.
\end{proof}

\begin{theorem}
[Completeness of \textbf{B}] \label{completudsecB} Let $\Gamma \cup \Delta$  be a finite nonempty subset of $For_2$. If $\models _{\mathbf{B}_3}\Gamma \Rightarrow \Delta$  then $\Gamma \Rightarrow \Delta$  is provable in \textbf{B} without using the Cut rule. In particular, if $\Gamma \models _{\mathbf{B}_3}\alpha$ then $\Gamma \Rightarrow \alpha$  is provable in~\textbf{B}.
\end{theorem}
\begin{proof}
Assume that $\models _{\mathbf{B}_3}\Gamma \Rightarrow \Delta $. Then, by Proposition~\ref{BextensionC} and Theorem~\ref{compleCPL}, it follows that  $\Gamma \Rightarrow
\Delta $  is provable  in the $\{\lnot,\wedge\}$-fragment of  \textbf{C}.  If $var\left( \Delta \right) \subseteq var\left( \Gamma \right) $ then, by Lemma \ref{variablesincluidasB},  the sequent $\Gamma \Rightarrow \Delta $ is provable  in \textbf{B} without using the Cut rule. On the other hand, if  $var\left(
\Delta \right) \varsubsetneq var\left( \Gamma \right) $, then by Lemma~\ref{variablesnoincluidasB}, $\vDash _{\mathbf{B}_3}\Gamma \Rightarrow \Delta' $,
for some set $\Delta'\subset \Delta $ such that $var\left( \Delta'\right) \subseteq var\left( \Gamma \right) $. 
By Proposition~\ref{BextensionC} and Theorem~\ref{compleCPL} again, it follows that  $\Gamma \Rightarrow
\Delta'$  is provable  in the $\{\lnot,\wedge\}$-fragment of  \textbf{C}. Using again  Lemma~\ref{variablesincluidasB},  the sequent $\Gamma \Rightarrow \Delta'$ is provable  in \textbf{B} without using the Cut rule. Finally,
by applying the structural rule $\Rightarrow W$ several times we obtain  a derivation of  $\Gamma \Rightarrow \Delta$ in \textbf{B} without using the Cut rule.
\end{proof}

\begin{corollary}
[Cut elimination for \textbf{B}] \label{cut-freeB} Let $\Gamma \cup \Delta $ be a finite nonempty set of formulas in $For_2$.
If the sequent $\Gamma \Rightarrow \Delta$ is provable in \textbf{B} then there is a cut-free derivation of it in  \textbf{B}.
\end{corollary}

\begin{proof}
Suppose that $\Gamma \Rightarrow \Delta$ is provable in \textbf{B}. By Theorem~\ref{correccionseqB}, $\models _{\mathbf{B}_3}\Gamma \Rightarrow \Delta$. Then, by Theorem~\ref{completudsecB}, there is  a cut-free  derivation of  $\Gamma \Rightarrow \Delta$ in \textbf{B} as desired. 
\end{proof}

\section{Concluding Remarks}

In this paper a cut-free sequent calculi for the  $\{\lnot,\vee\}$-fragment of Bochvar's logic,
as well as a cut-free sequent calculi for the  $\{\lnot,\wedge\}$-fragment of Halld\'en's logic, were proposed.
In the former calculus,  conjunction and implication are derived connectives, while disjunction and implication are derived connectives in the latter. The main feature of both calculi is that they are obtained by imposing provisos to the rules of the respective fragments of a well-known sequent calculus for classical propositional logic. The signature for each calculus was choosen in order to keep as close as possible to the respective fragment of classical logic. Observe that both  $\{\lnot,\vee\}$ and $\{\lnot,\wedge\}$-fragments are adequate, that is, they can express all the other (classical) connectives. 

Thus, concerning the calculus for the $\{\lnot,\vee\}$-fragment of  Halld\'en's logic, the only change required with respect to the calculus for  the respective fragment of classical logic was the inclusion of a proviso in the introduction rule for negation on the left side of the sequent. As a consequence of this, a proviso appear in the (derived) introduction rules for conjunction and implication on the left side of the sequent.

In the calculus for the  $\{\lnot,\wedge\}$-fragment of  Bochvar's logic, the situation is entirely symmetrical: the restriction was imposed to the introduction rule for negation on the right side, and so this restriction also applies to  the introduction rules for disjunction and implication on the right side (both are derived rules).  In this manner, the existing relationship  between classical logic and both logics  became explicit through restrictions on the rules for the logical connectives.

Since these two logic of nonsense are related to classical logic in such particular way,  the {\em ad hoc} definition of sequent calculi presented here, which exploit these particularities, seems to be justified. However, it would be interesting to compare the cut-free sequent calculi introduced here with the ones which could be obtained by applying general techniques such as those proposed in~\cite{baa:fer:zac:93,avr:ben:kon:07,vol:mar:cal:12}.

As a future research, we plan to extend the calculi to  the full language of both logics. Clearly the resulting calculi will not be so simple and symmetrical because of  the subtleties of the `meaningful' connectives and their relationship with the other connectives.

\

\

\noindent {\bf Acknowledgements:}  We would like to thank the anonymous referees for their extremely useful  comments on an earlier draft, which have helped to improve the paper. The first author was financed
by FAPESP (Brazil), Thematic Project LogCons 2010/51038-0 and by
an individual research grant from The National Council for
Scientific and Technological Development (CNPq), Brazil.

\nocite{*}
\bibliographystyle{eptcs}
\bibliography{refer}
\end{document}